\documentclass{article}

\usepackage{amssymb}
\usepackage{amsmath}
\usepackage{amsthm}

\theoremstyle{definition}
\newtheorem{example}{Example}
\newtheorem{definition}{Definition}

\usepackage{graphicx}

\usepackage[capitalise]{cleveref}
\usepackage{paralist}

\newtheorem{problem}{Problem}
\newtheorem{corollary}{Corollary}
\newtheorem{theorem}{Theorem}

\newtheorem{proposition}{Proposition}

\newtheorem{lemma}{Lemma}

\usepackage{times}
\title{Robust Draws in Balanced Knockout Tournaments}
\author{Krishnendu Chatterjee$^\dag$
\qquad Rasmus Ibsen-Jensen$^\dag$
\qquad Josef Tkadlec$^\dag$  \\[5pt]
{\normalsize $\strut^\dag$ IST Austria} 
\\ 
}
\begin{document}

\def\eps{\varepsilon}
\def\p{\mathcal{P}}
\def\wp{\operatorname{wp}} 
\def\mwp{\operatorname{mwp}} 
\def\wpe{\operatorname{wp}_\eps} 
\def\en{\mathbb{N}}
\def\er{\mathbb{R}}
\def\wh{\widehat}

\onecolumn

\maketitle
\begin{abstract}
Balanced knockout tournaments are ubiquitous in sports 
competitions and are also used in decision-making and 
elections.
The traditional computational question, that asks to compute 
a draw (optimal draw) that maximizes the winning probability 
for a distinguished player, has received a lot of attention. 
Previous works consider the problem where the pairwise 
winning probabilities are known precisely, while we study 
how robust is the winning probability with respect to small errors in the 
pairwise winning probabilities.
First, we present several illuminating examples to establish: 
(a)~there exist deterministic tournaments (where the pairwise 
winning probabilities are~0 or~1) where one optimal draw is 
much more robust than the other; and 
(b)~in general, there exist tournaments with slightly suboptimal 
draws that are more robust than all the optimal draws.
The above examples motivate the study of the computational 
problem of robust draws that guarantee a specified winning probability. 
Second, we present a polynomial-time algorithm for approximating
the robustness of a draw for sufficiently small errors in pairwise
winning probabilities, and obtain that the stated computational
problem is NP-complete.
We also show that two natural cases of deterministic tournaments 
where the optimal draw could be computed in polynomial time also 
admit polynomial-time algorithms to compute robust optimal draws.
\end{abstract}

\section{Introduction}\label{sec:intro}
\smallskip\noindent{\em Balanced knockout tournaments (BKTs).}
A Balanced knockout tournament (BKT) consists of $N=2^n$ players (for some positive integer $n$) and is played in $n$ rounds. In every round each remaining player plays a (win/lose) game with some other remaining player and the loser is eliminated (removed from the tournament). In the end, only one player remains and is declared the winner. 
The whole process can be visualised on a balanced binary tree (over $N$ leaves) with players starting at the leaves and winners advancing to the parent.
Such tournaments are used in many sports such as tennis, elections and other decision making processes.

\smallskip\noindent{\em Draws in BKT.}
For a given set of $N$ players, the tournament can be played in many different 
ways based on how we draw the players in the first round (i.e., on which leaf they are placed in the balanced binary tree). 
Even if we consider the draws that correspond to isomorphic labeled binary trees as equivalent, 
the number of draws grows rapidly in $N$, precisely, 
there are $\frac{N!}{2^{N-1}}$ draws.
As a numerical example, for a Grand Slam tournament in tennis 
with $N=128$, the number of distinct draws is at least $10^{177}$ 
and even if we rank the players and require that, for any $k$, the 
top $2^k$ of them do not meet until the last $k$ rounds,
we have at least $10^{144}$ distinct draws.

\smallskip\noindent{\em Computational problems.}
The traditional computational problem that has been extensively studied  
is related to obtaining draws that are most favorable for a 
distinguished player~\cite{VAS09,AGMM,SW11b}.
The input consists of an $N\times N$ comparison matrix $P$ that specifies 
the probabilities $P_{ij}$ of player $i$ beating player $j$ in a single match,
and a distinguished player~$i^*$.
The input, along with a draw, defines the winning probability that the   
distinguished player wins the whole tournament. If the comparisons 
matrix contains 0's and 1's only, a draw gives a unique winner.
The {\em probabilistic (resp. deterministic) tournament fixing problem} 
(PTFP) (resp. TFP) asks whether there exists a draw such that the 
winning probability for the distinguished player~$i^*$ is at least 
$q$ (resp., player~$i^*$ is the winner).
The TFP and PTFP problems are desired to be hard, as otherwise 
the tournament could be manipulated 
by choosing the draw that favors a specific player. 

\smallskip\noindent{\em Previous results.}
It was shown in~\cite{VAS09} that given an input comparison matrix 
and a draw, the winning probability for a player can be computed by a 
recursive procedure in $O(N^2)$ time (i.e., in polynomial time). 
Since a candidate draw for the PTFP problem is a polynomial witness, 
an NP upper bound follows for both the PTFP and the TFP problem.
An NP lower bound for the PTFP problem was shown in~\cite{VAS09,W10}, and 
finally the TFP problem was shown to be NP hard in~\cite{AGMM}.
Hence both the PTFP and the TFP problems are NP-complete.
Moreover, in~\cite{AGMM} two important special cases have been
identified (namely, constant number of player types, and 
linear ordering among players with constant number of exceptions)
where all the winning draws for player~$i^*$ can be computed in 
polynomial time for deterministic tournaments.

\smallskip\noindent{\em Robustness question.}
The previous works focused on the computational problems 
when the pairwise winning probabilities are precisely known.
However, in most practical scenarios, the winning probabilities 
in the comparison matrix are only approximation of the real probabilities.
Indeed, in all typical applications, either the probabilities are obtained 
from past samples (such as games played) or uniformly selected from a 
small subset of samples (such as in elections). 
The probabilities obtained in these ways are always at best an 
approximation of the real probabilities and subject to small errors.
This leads to the natural question about robustness (or sensitivity) 
of optimal draws for probabilistic tournaments (or winning draws in 
deterministic tournaments) in the presence of small errors in the 
pairwise winning probabilities.
Formally, given a comparison matrix $P$, and a small error term 
$\eps$, we consider {\em all} comparison matrices $P'$ where each
entry differs from $P$ by at most $\eps$.
We refer to the above set as the $\eps$-perturbation matrices of $P$.
For a draw, we consider the {\em drop} of the draw as the difference 
between the winning probability for the distinguished player $i^*$ for $P$
and the infimum of the winning probability of the $\eps$-perturbation
matrices.

\smallskip\noindent{\em Our results.}
We study the computational problems related to robust draws in TFP and PTFP. 
Our main contributions are: 
\begin{compactenum}

\item {\em Examples.} 
We present several illuminating examples for TFP and PTFP 
related to robustness. 
First, we show that there exist deterministic tournaments where
for one winning (or optimal draw) the drop is about $\frac{N \cdot \eps}{2}$,
whereas for another winning draw the drop is about $\log N \cdot \epsilon$.
In other words, one winning draw is much more robust than the other. 
Second, we show that there exist probabilistic tournaments with 
suboptimal draws that are more robust than all the optimal draws.
Motivated by the above examples we study the computational 
problem of robust draws, i.e. determining the existence of draws that 
guarantee a specified winning probability with a drop below a specified 
threshold. 

\item {\em Algorithm.} 
We present a polynomial-time algorithm to approximate the drop for 
a given draw for small $\eps>0$ (informally, for $\eps$ where
higher order terms of $\eps$ such as $\eps^2, \eps^3$ 
etc. can be ignored).

\item {\em Consequences.} 
Our algorithm has a number of consequences.
First, it establishes that the computational problem of robust 
draws for small $\epsilon$ is NP-complete.
Note that while the PTFP and TFP are existential questions (existence
of a draw), the robustness question has a quantifier alternation (existence 
of a draw such that for all $\epsilon$-perturbation matrices the drop is small), 
yet we match the complexity of the PTFP and TFP problem.
Second, our polynomial-time algorithm along with the result of~\cite{AGMM}
implies that for the two natural cases of~\cite{AGMM}, the most robust 
winning draw (if one exists) is polynomial-time computable.
\end{compactenum}
Conference version of the paper appeared in~\cite{KIT_TechRpt}.

\smallskip\noindent{\em Significance as risk-averse strategies.}
As mentioned, BKTs are often used as sports tournaments. 
After a draw is fixed, the winning probabilities determine 
the betting odds.
Since the comparison matrix can only be approximated,
a risk-averse strategy (as typically employed by humans) 
corresponds to the notion of robustness.
The notion of robustness has been studied in many different contexts,
such as for sensitivity analysis in MDPs~\cite{Puterman,FV97} as 
well as for decision making in markets~\cite{Rockafellar}.
Our algorithm provides a risk-averse approximation for balanced knockout
tournaments, for low levels of uncertainty in the probabilities.


\subsection{Related Work}
The most related previous works are: 
(a)~\cite{VAS09}, who  showed that for a fixed draw the 
probability distribution over the winners can be computed in $O(N^2)$ time;
and (b)~\cite{AGMM}, who determined the complexity of TFP and 
found special cases with polynomial-time algorithms.
Besides that, \cite{W10}  
identified various sufficient conditions for a player to be a winner
of a BKT; \cite{SW11a} considered the case when 
weak players can possibly win a BKT; \cite{SW11b}  
studied the conditions under which the tournament can be 
fixed with high probability. 
The problem of fair draws was studied in~\cite{VS11}; and the 
problem of determining the winner in unbalanced voting trees was considered
in~\cite{Lang07,Lang12}.
If only incomplete information on the preferences is known, then computing
the winner with various voting rules has been studied in~\cite{XC11,Aziz12}.
The problem of checking whether a round-robin competition can be won 
when all the matches are not yet played has been studied in~\cite{KP04,GM02}.

As compared to the existing works we consider the problem of robustness
for TFP and PTFP which has not been studied before.

\section{Preliminaries}
In this section we present the formal definitions and previous results.

\subsection{Definitions}
For the basic definitions we very
closely follow the notations of~\cite{AGMM}.

\begin{definition} [Comparison matrices]
For $N\in\en$, let $[N]:=\{1,\dots,N\}$. Consider $N$ players numbered 1 to $N$. A \textit{comparison matrix\/} is an $N\times N$ matrix $P$ such that for all $1\leq i\neq j\leq N$ we have $0\leq P_{ij}\leq 1$ and $P_{ij}+P_{ji}=1$. The entry $P_{ij}$ of the comparison matrix expresses the probability that if players $i$, $j$ play a match, player $i$ wins. Note that the entries $P_{ii}$ are not defined.
A comparison matrix is {\em deterministic} if $P_{ij}\in \{0,1\}$ for all $i,j \in [N]$ with $i\neq j$.
\end{definition}

\begin{definition} [Draws]
Let $N=2^n$ for some integer $n$. For $\sigma$ a permutation of $[N]$, an \textit{$n$-round ordered balanced knockout tournament} $T([N],\sigma)$ is a binary tree with $N$ leaf nodes labelled from left to right by $\sigma$. All ordered balanced knockout tournaments that are isomorphic to each other 
are said to have the same \textit{draw}. They are represented by a single (unordered) \textit{balanced knockout tournament} $T([N],\sigma)$, where $\sigma$ is again a permutation of $[N]$. The set of all draws is denoted by $\Sigma$.
\end{definition}

\begin{definition} [Complete tournaments]
A \textit{complete tournament} $C(P,\sigma)$ is a balanced knockout tournament $T([N],\sigma)$ together with an $N\times N$ comparison matrix $P$.
A complete tournament $C(P,\sigma)$ is called \textit{deterministic\/} if the comparison matrix $P$ is deterministic. 
The complete tournament $C(P,\sigma)$ is conducted in the following fashion. If two nodes with labels $i$, $j$ have the same parent in $T([N],\sigma)$ then players $i$, $j$ play a \textit{match}. The winner then labels the parent (i.e. the parent is labeled by $i$ with probability $P_{ij}$ and by $j$ otherwise). The \textit{winner of} $C(P,\sigma)$ is the player who labels the root node.
\end{definition}

\begin{definition} [Winning probabilities]
Given a complete tournament $C(P,\sigma)$, each player $i\in [N]$ has a probability, denoted $\wp(i,P,\sigma)$, 
of being the winner of $C(P,\sigma)$. This probability can be computed in time $O(N^2)$ via a recursive formulation~\cite{VAS09}. 
We denote by $\mwp (i,P):=\max_{\sigma\in\Sigma}\{\wp(i,P,\sigma)\}$ the maximum possible winning probability of $i$ in $C(P,\sigma)$ taken over all draws $\sigma\in\Sigma$. 
Given $P$, and $\delta\geq 0$, a draw $\sigma$ is called \textit{$\delta$-optimal} for player $i$ provided that 
$\wp(i,P,\sigma)\geq \mwp(i,P)-\delta$.
A draw is {\em optimal} if it is $\delta$-optimal for $\delta=0$.
\end{definition}

\begin{example}
Consider a complete 2-round tournament with comparison matrix $P$ and draw $\sigma=(1,2,3,4)$ as in Figure~\ref{fig:easy}.
\begin{figure}
\centering
\includegraphics[scale=0.8]{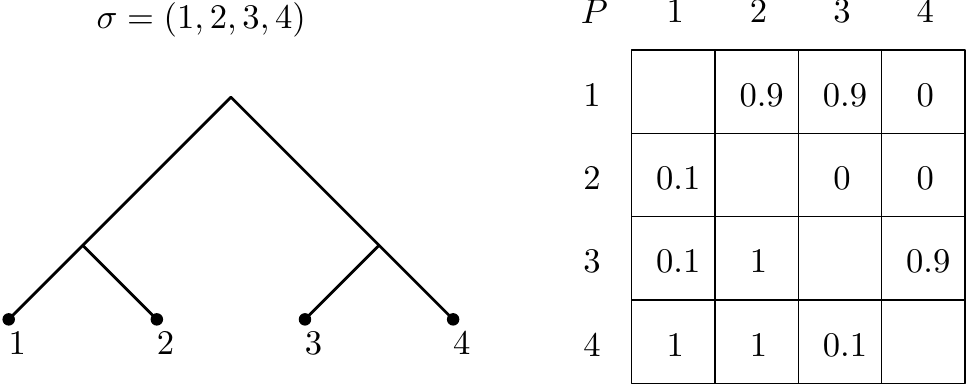}
\caption{A 2-round complete tournament with comparison matrix $P$. The draw $\sigma$ gives $\wp(1,P,\sigma)=0.9^3$.}
\label{fig:easy}
\vspace{-1em}
\end{figure}
Then the winning probabilities are
\begin{align*}
\wp(1,P,\sigma)&=0.729, \quad \wp(2,P,\sigma)=0,\\
\quad \wp(3,P,\sigma)&=0.171, \quad \wp(4,P,\sigma)=0.1
\end{align*}
and the draw $\sigma$ is optimal for player 1, as the draws $\sigma'=(1,3,2,4)$, $\sigma''=(1,4,2,3)$ give 
$\wp(1,P,\sigma')=\wp(1,P,\sigma'')=0$.
\end{example}

\subsection{Previous Results}
We now describe the main computational problems and previous results related
to them.

\begin{problem} [TFP -- Tournament Fixing Problem] 
Given a player set $N$, a deterministic comparison matrix $P$, and a distinguished player $i^*\in N$, 
does there exist a draw $\sigma$ for the player set $N$ for which $i^*$ is the winner of $C(P,\sigma)$?
\end{problem}

\begin{problem} [PTFP -- Probabilistic TFP] 
Given a player set $[N]$, a comparison matrix $P$, a distinguished player $i^*\in [N]$ and target probability $q\in[0,1]$, 
does there exist a draw $\sigma$ for the player set $[N]$ such that $\wp(i^*,P,\sigma) \geq q$?
\end{problem}

\begin{theorem} (\cite{AGMM,VAS09})
The TFP and the PTFP problems are NP-complete.
\end{theorem}

\section{Robustness: Examples and Questions}

In practical applications of PTFP and TFP, the exact values of the entries in the comparison matrix are not known, 
but only approximations are obtained. 
Hence it is of interest to find out how sensitive is the winning chance in optimal draws (or near optimal draws) 
with respect to minor changes in the comparison matrix.
We first present the basic definition and then our examples. 
Finally, we present the computational problems.

\begin{definition}[$\eps$-perturbation and the $\eps$-worst drop]
Given a comparison matrix $P$ and $\eps>0$, an $\eps$-perturbation of $P$ is any comparison matrix $P'$ 
such that $|P_{ij}-P'_{ij}|\leq \eps$ for all $1\leq i\neq j\leq N$. 
The set of all $\eps$-perturbations of $P$ is denoted $\p(P,\eps)$. 
Given a complete tournament $C(P,\sigma)$, a distinguished player $i^*$, and $\eps>0$, 
define the \textit{$\eps$-guaranteed winning probability} $\wpe(i^*,P,\sigma)$ by
$$\wpe(i^*,P,\sigma)=\inf_{P'\in \p(P,\eps)} \{\wp(i^*,P',\sigma)\}
$$
and the \textit{$\eps$-worst drop}, denoted $d_{\eps}(i^*,P,\sigma)$, by  
$\wp(i^*,P,\sigma)-\wpe(i^*,P,\sigma)$.
The smaller the drop $d_{\eps}(i^*,P,\sigma)$ the more robust (or more risk-averse) is the draw $\sigma$.
\end{definition}

\subsection{Examples}
We consider the robustness problem for sufficiently small $\eps>0$.
Intuitively, ``sufficiently small'' will allow us to ignore all higher order terms of $\eps$ such as $\eps^2, \eps^3,\dots$ A formal definition comes at the end of this section.



First we construct a deterministic tournament with a unique and nonrobust optimal draw.
We will use Proposition~\ref{prop1} in Proposition~\ref{prop2}. For brevity, the distinguished player will always be the first one, i.e. $i^\star=1$.

\begin{proposition}[Only Nonrobust Optimal Draws]\label{prop1}
For any $2^n=N\in\en$, there exists a deterministic $n$-round tournament (called a {\em hard $n$-round tournament} and denoted $\mathcal{H}_n$) with a 
comparison matrix $P$ such that for every draw $\sigma$ which makes player 1 win 
we have
$$\wpe(1,P,\sigma)\leq 1-(N-1)\eps +\eps^2\cdot Q(\eps)$$
for some
integer polynomial $Q(\eps)$.
\end{proposition}
Intuitively, the polynomial $Q$ stores higher order terms (from $\eps^2$ on). The tournament, illustrated in Figure~\ref{fig:hard}, has the property that exactly one draw makes player~1 win and if any single match changes outcome then player~1 loses.

\begin{proof} 
We use a construction from \cite[Lemma 1]{AGMM}. 
For convenience, let us repeat it here: Start with player 1. At each iteration, each player $a$ \textit{spawns\/} a player $b$ directly to their right. This is repeated until $2^n$ players are present. In the pairwise comparison matrix $P$, each player beats all players to their left except for the one that spawned them. We show that this is a hard $n$-round tournament.

\begin{figure}
\centering
\includegraphics{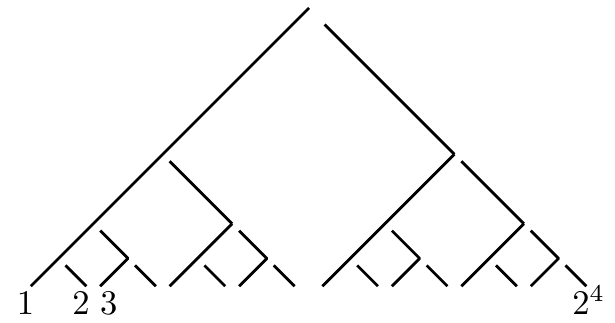}
\caption{A 4-round hard tournament $\mathcal{H}_4$. Each broken line connects nodes that are eventually labelled by the same player. Except for the depicted exceptions, each player beats all players to their left.}
\label{fig:hard}
\end{figure}

Number the players from left to right. \cite[Lemma 1]{AGMM} states that in the draw $\sigma=(1,2,\dots,2^n)$, player 1 wins. We prove that this is the only draw that makes player 1 win.

We proceed by induction on $n$. For $n=1$ the statement is trivial. Assume it is true for some fixed $n\in\en$ and take a hard $(n+1)$-round tournament with some draw making player 1 win. Call players with numbers less than $2^n+1$ \textit{small} and those with numbers greater than $2^n+1$ \textit{big}. Note that every big player beats every small player.

The players $1$ to $2^{n+1}$ are divided into left half and right half. If there exists a big player who is not in the same half as player $2^n+1$, then that half is won by a big player and hence the whole tournament is not won by player 1, a contradiction. Hence the division into halves is (in some order) $\{1,\dots,2^n\}$, $\{2^n+1,\dots,2^{n+1}\}$ and since both the subtournaments are isomorphic to the hard $n$-round tournament, we are done by induction.

Now we define an $\eps$-perturbation of $P$ that gives a small winning chance for player 1 in the draw $\sigma$. Let $P'$ be the $\eps$-perturbation where we adjusted every (0,1) match into an $(\eps,1-\eps)$ match.
Denote the probability that the winner of a $k$-round subtournament in $C(P,\sigma)$ actually wins their first $k$ rounds in $C(P',\sigma)$ by $p_k$ (this is possible since all subtrees of the same size are isomorphic). Then $$p_{k+1}=p_k\cdot \left[p_k\cdot (1-\eps) + (1-p_k)\cdot\eps\right],$$
because the ``left'' original winner has to win their first $k$ rounds and then
\begin{itemize}
\item either beat the original winner of the right subtournament, provided that he won,
\item or upset the real winner of the right subtournament, if he didn't.
\end{itemize}
Since $p_1=1-\eps$, we can express $p_n$ recursively to obtain a polynomial in $\eps$. By straightforward induction we verify that its degree equals $2^n-1$ and that the coefficient by the linear term equals $2^n-1=N-1$. Thus $$\wpe(1,P,\sigma)\leq \wp(1,P',\sigma)= 1-(N-1)\eps + \eps^2Q(\eps)$$ for some polynomial $Q(\eps)$.
\end{proof}

Next we present a concrete numerical example to illustrate Proposition~\ref{prop1}.

\begin{example} There exists a deterministic 6-round tournament with a comparison matrix $P$ such that any draw $\sigma$, which makes player 1 win, satisfies
\begin{align*}
\wp_{0.01}(1,P,\sigma)&<0.54, \quad\wp_{0.05}(1,P,\sigma)<0.07,\\
\quad \wp_{0.1}(1,P,\sigma)&<0.02.
\end{align*}
By Proposition~\ref{prop1}, the 6-round hard tournament with comparison matrix $P$ has only one draw $\sigma=(1,\dots,64)$ that makes player 1 win. Let $P'$ be the $\eps$-perturbation where we adjusted every (0,1) match into an $(\eps,1-\eps)$ match. We recursively compute $\wp(1,P',\sigma)=1-63\eps +\dots-2\ 147\ 483\ 648 \eps^{63}$
and plug in $\eps\in\{0.01,0.05,0.1\}$ to get numbers less than $0.54$, $0.07$, $0.02$, respectively.
\end{example}

We now use Proposition~\ref{prop1} to show that there exist deterministic tournaments 
where for one winning draw the drop is approximately $\frac{N\eps}{2}$, whereas for
a different winning draw the drop is approximately $\eps(2\log N-1)$, for sufficiently
small $\eps$.
In other words, some winning draws can be much more robust than other winning draws.

\begin{proposition}[Robust and Nonrobust Optimal Draws]\label{prop2} For any $2^n=N\in\en$, $n\geq 2$, there exists a deterministic $n$-round tournament (called an \textit{unbalanced $n$-round tournament} and denoted by $\mathcal{U}_n$) with a comparison matrix $P$ and draws $\sigma$, $\sigma'$ such that $\wp(1,P,\sigma)=1=\wp(1,P,\sigma')$ and
\begin{align*}
\wpe(1,P,\sigma) &\leq 1- \eps\cdot N/2 +\eps^2\cdot Q(\eps),\\
\wpe(1,P,\sigma')&> 1- \eps\cdot (2n-1)
\end{align*}
for some integer polynomial $Q(\eps)$.
\end{proposition}
\begin{proof}  Take an $(n-1)$-round hard tournament with players $1,2,\dots,2^{n-1}$. Add $2^{n-1}$ more players numbered $2^{n-1}+1$ through $2^n$ who all lose to player 1 but beat all the players 2 through $2^{n-1}$ (see Figure~\ref{fig:unbalanced}). Then we claim that $\sigma=(1,2,\dots,2^n)$ and $\sigma'=(2^n,2,3,\dots,2^n-1,1)$ fulfil the statement of the proposition.

\begin{figure}
\centering
\includegraphics[scale=1]{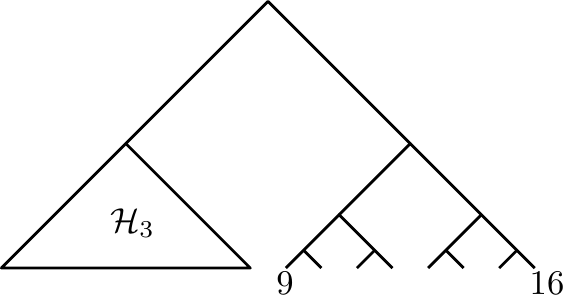}
\caption{An unbalanced $4$-round tournament $\mathcal{U}_4$. The players in the right subtree beat everyone from the left subtree but the player 1.}
\label{fig:unbalanced}
\end{figure}

Clearly, player 1 wins both $C(P,\sigma)$ and $C(P,\sigma')$.

To upper bound $\wpe(1,P,\sigma)$ it is enough to find one $\eps$-perturbation that reduces the winning chance of player 1 sufficiently. Let $P'$ again be the comparison matrix obtained by replacing each $(0,1)$ match by an $(\eps,1-\eps)$ match. By Proposition~\ref{prop1}, the probability of player 1 winning her first $n-1$ rounds is $p_{n-1}= 1-(2^{n-1}-1)\eps+\eps^2Q'(\eps)$ for some polynomial $Q'(\eps)$. The probability of winning the final is then always $1-\eps$ which altogether gives
\begin{align*}
\wpe(1,P,\sigma)&\leq  [1-(2^{n-1}-1)\eps +\eps^2Q'(\eps)]\cdot (1-\eps) \\
&=1-2^{n-1}\eps +\eps^2Q(\eps)
\end{align*}
for some polynomial $Q(\eps)$.

For $\wpe(1,P,\sigma')$, no matter which $P''\in\p(P,\eps)$ we take, players $2^n$ and $1$ win their first $n-1$ matches with probability at least $(1-\eps)^{n-1}$ each and then player 1 beats player $2^n$ with probability at least $1-\eps$ again, giving
$$\qquad\wpe(1,P,\sigma')\geq (1-\eps)^{2n-1}> 1-(2n-1)\eps.\qquad \qedhere$$
\end{proof}

We now illustrate Proposition~\ref{prop2} with a numerical example.

\begin{example} There exists a deterministic 6-round tournament with a comparison matrix $P$ and draws $\sigma$, $\sigma'$ such that $\wp(1,P,\sigma)=1=\wp(1,P,\sigma')$ and
\begin{align*}
\wp_{0.01}(1,P,\sigma)< 0.73 &< 0.89<\wp_{0.01}(1,P,\sigma'),\\
\wp_{0.05}(1,P,\sigma)< 0.23 &< 0.56<\wp_{0.05}(1,P,\sigma'),\\
\wp_{0.1 }(1,P,\sigma)< 0.08 &< 0.31<\wp_{0.1 }(1,P,\sigma').
\end{align*}
Indeed, take the unbalanced 6-round tournament with comparison matrix $P$ and the draws $\sigma=(1,2,\dots,64)$, $\sigma'=(64,2,3,\dots,63,1)$. By Proposition 2, player one wins both $C(P,\sigma)$ and $C(P,\sigma')$. Consider the $\eps$-perturbation $P'$ of $P$ that replaces each $(0,1)$ match by an $(\eps,1-\eps)$ match. Then for any $\eps>0$ we have
\begin{align*}
\wpe(1,P,\sigma)&\leq \wp(1,P',\sigma)= p_5\cdot (1-\eps)= \\
&=1-64\eps+\dots +2\ 147\ 483\ 648 \eps^{64} 
\end{align*}
which, after plugging in $\eps\in\{0.01, 0.05, 0.1\}$, gives the desired inequalities.

On the other hand, the proof of the Proposition~\ref{prop2} implies that for any $\eps>0$ we have $\wpe(1,P,\sigma')\geq (1-\eps)^{11}$ which gives the remaining inequalities.
\end{example}

\smallskip\noindent{\em Probabilistic setting.}
Proposition~\ref{prop1} and Proposition~\ref{prop2} concern the deterministic
setting.
In a deterministic setting, if a draw is not optimal (i.e., winning), then the player
loses with probability~1 (hence the notion of $\delta$-optimality, for $\delta>0$
is not relevant).
We now consider the probabilistic setting and
 show that there are combinations of $\delta$ and $\eps$ such that a certain $\delta$-optimal 
draw gives a better $\eps$-guarantee than any optimal draw.
In other words, near optimal draws can be more robust than all optimal draws.

\begin{proposition}\label{prop3}
There exists a tournament with comparison matrix $P$ and a distinguished player $i$ and there exist $\delta,\eps>0$ and a $\delta$-optimal draw $\sigma'$ such that for all optimal draws $\sigma$ we have $\wpe(i,P,\sigma)<\wpe(i,P,\sigma')$.
\end{proposition}

For proving Proposition~\ref{prop3} we use an auxiliary lemma.

Let a {\em big vs. small tournament} with players $[2^n]$ have the comparison matrix $P$ such that each of the players in $\{2^{n-1}+1,\dots,2^n\}$ (``big'' players) beats each of the players in $[2^{n-1}]$ (``small'' players) with the same probability $p\in(0.5,1)$.

\begin{lemma}[Mixed draws]
\label{lemm1} 
The draws that maximize the probability of a big player winning the big vs. small tournament (the \textit{big-optimal} draws) are precisely those that pair up big players and small players in the first round (so-called \textit{mixed} draws).
Moreover, for $p=0.5+\eps$, the probability $b(P)$ of a big player winning in such a draw equals $b(P)=0.5 + \frac12(n+1)\eps +\eps^2\cdot Q(\eps)$ for some polynomial $Q(\eps)$.
\end{lemma}

For proving \cref{lemm1} we will need two more lemmas.

\begin{lemma}
\label{reorder} Let $C(P,S)$ be a complete $(n+2)$-round tournament. Denote the four $n$-round subdraws it consists of by $A$, $B$, $C$, $D$, i.e. $S=(A,B,C,D)$. Suppose that $a\geq b\geq c\geq d$ and either $b>c$ or both $a>b$ and $c>d$. Let $S'=(A,D,B,C)$. Then $b(P,S')>b(P,S)$.
\end{lemma}

\begin{proof} This is straightforward. After expressing both $b(P,S')$ and $b(P,S)$ in terms of $a$, $b$, $c$, $d$, cancelling equal terms and dividing by $p(2p-1)$, the inequality $b(P,S')>b(P,S)$ is equivalent to $(1-p)(a-c)(b-d)>0$.
\end{proof}

\begin{lemma}
\label{core}
Given a big-optimal draw $S$ of $k$ ones and $2^n-k$ zeros, the following is true:
\begin{itemize}
\item The number of 1's in the left half and the right half differs by at most 1.
\item There is a unique way to arrange the $k$ 1's in an optimal draw (up to isomorphism).
\item This draw majorizes (entrywise) the optimal draws with fewer 1's.
\end{itemize}
\end{lemma}
\begin{proof} By induction on $n$. The statements are clearly true for $n=0$ and $n=1$. For the induction step, split the draw into four subdraws of the same size. By induction and without loss of generality, the numbers of ones $s$, $t$, $u$, $v$ in these subdraws satisfy $s\in\{t,t+1\}$, $u\in\{v,v+1\}$, $s\geq u$. If $s=t+1$ then casework and \cref{reorder} imply that either the numbers are $(n,n-1,n,n)$ or $(n,n-1,n,n-1)$ or $(n,n-1,n-1,n-1)$. If $s=t$ then \cref{reorder} implies that either the numbers are $(n,n,n,n)$ or $(n,n,n,n-1)$. Hence the first statement follows. For the second statement, note that the first and fifth option are isomorphic and the last four options exactly cover four possible sums $4n-3$ to $4n$. The third statement follows from
\begin{align*}
(n,n,n,n)&\succ (n,n-1,n,n)\succ (n,n-1,n,n-1)\succ \\
&\succ(n,n-1,n-1,n-1)\\
&\succ (n-1,n-1,n-1,n-1)
\end{align*}
and induction (here $\succ$ denotes majorization entrywise).
\end{proof}

\begin{proof}(of Lemma~\ref{lemm1})
Given any draw $S=(L,R)$ with left half $L$ and right half $R$, the chance of a big player winning $S$ is 
$$s= \ell r + p\ell(1-r) + p r(1-\ell) = p(\ell+r) + (1-2p)\ell r.
$$
Note that the right-hand side is linear and symmetric in both $\ell$ and $r$. We claim that it is also increasing in each of them. Indeed, rewriting it as $\ell(p+r-2pr)+pr$ it suffices to check that $p+r>pr+rp=2pr$.

Mark big players by 1 and small ones by 0. Note that if a subdraw $S$ is (strictly) majorized by a subdraw $S'$ (entrywise) then $b(S')> b(S)$.

Now we finish the proof of the first part of \cref{lemm1} easily.
Have a look at the big-optimal draw. The total number of ones is even so by Lemma~\ref{core} they have to divide evenly between left and right half. By induction we are done.

For the second part, denote by $e_k$ the probability of an even player winning in any $k$-round optimal tournament. Then $e_1=0.5+\eps$ and
\begin{align*}
e_{k+1}&=2e_k(1-e_k)(0.5+\eps) + e_k^2=\\
      &=e_k + 2e_k(1-e_k)\eps.
\end{align*}
The statement now follows by straightforward induction.
\end{proof}

\begin{proof} {\em (of Proposition~\ref{prop3})} Consider an $(n+1)$-round tournament with players $[2^{n+1}]$. Call players $\{2,3,\dots,2^n\}$ \textit{small}, players $2^n+1,\dots,2^n+2^{n-1}$ \textit{medium} and players $2^n+2^{n-1}+1,\dots,2^{n+1}$ \textit{big}. For $s$ small, $m$ medium, and $b$ big, let the comparison matrix $P$ satisfy:
\begin{itemize}
\item $P_{1s}=1$, $P_{1m}=0.4$, $P_{1b}=0.6$,
\item $P_{sm}=P_{sb}=0$,
\item $P_{mb}=0.5-\frac12\eps$.
\end{itemize}
We claim that:
\begin{enumerate}
\item The draws optimal for player 1 are those that have all the small players in one half and the medium players paired up with the big players in the other half.
\item For any optimal draw $\sigma$, the draw $\sigma'=(1,\dots,2^n)$ satisfies
$ \wp_{\eps}(1,P,\sigma)<\wp_{\eps}(1,P,\sigma')$
for all sufficiently small $\eps>0$.
\end{enumerate}
Indeed:

1. Let $\sigma$ be a draw of the desired form. By Lemma~\ref{lemm1}, 
we have 
\begin{align*}
\wp(1,P,\sigma)= & \left[0.5+\frac14(n+1)\eps +\eps^2Q(\eps)\right]\cdot 0.6 + \\
 & \left[0.5-\frac14(n+1)\eps -\eps^2Q(\eps)\right]\cdot 0.4 \\
 &= 0.5+\frac1{20}(n+1)\eps + \eps^2R(\eps).
\end{align*}
for some polynomials $Q(\eps), R(\eps)$.

Since there are $2^n-1<2^n$ small players, player 1 has to face a medium or big player in the final (provided he made it that far). 
Now let $\tau$ be a draw optimal for player 1 and suppose it contains a medium or big player in the left half. Then the player 1 has to face at least one more such player before the final. Hence his winning chance is at most $0.6^2<\wp(1,P,\sigma)$ for all sufficiently small $\eps>0$, and $\tau$ is not optimal. All the small players are therefore in the left half. For the right half, we want to maximize the probability of a big player winning it. By Lemma~\ref{lemm1}, this is accomplished by the mixed draws.

2. Let $\sigma$ be optimal. Note that since $\sigma$ and $\sigma'$ only differ in the right subtree, the inequality $$\wp_{\eps}(1,P,\sigma)<\wp_{\eps}(1,P,\sigma')$$
holds if and only if the minimum possible probability $b_\eps(\sigma)$ of a big player winning the right half of $C(P^*,\sigma)$, taken over all $P^*\in\mathcal{P}(P,\eps)$, is less than the minimum possible probability $b_\eps(\sigma')$ of a big player winning the right half of $C(P',\sigma')$, taken over all $P'\in\mathcal{P}(P,\eps)$.

For $\sigma$, consider the $\eps$-perturbation that replaces each $(0.5-\frac12\eps,0.5+\frac12\eps)$ match by a $(0.5+\frac12\eps,0.5-\frac12\eps)$ match. Then by Lemma~\ref{lemm1} we have $$b_{\eps}(\sigma)=0.5 -\frac14(n+1)\eps - \eps^2Q(\eps)$$
for some polynomial $Q(\eps)$.

On the other hand, for all $P'\in\mathcal{P}(P,\eps)$ a big player certainly wins the right half of the right subtree of $C(P',\sigma')$ (since that subtree only consists of big players), hence
$$b_\eps(\sigma')\geq \left(0.5+\frac12\eps\right) - \eps = 0.5 -\frac12\eps.$$
For $n\geq 2$ and all sufficiently small $\eps>0$ we then get the desired $b_\eps(\sigma)<b_\eps(\sigma')$ and thus in turn $\wp_{\eps}(1,P,\sigma)<\wp_{\eps}(1,P,\sigma')$. \hfill$\qedhere$
\end{proof}

We again illustrate this with a numerical example.
\begin{example} Consider a 6-round tournament with players $[64]$. Call players $\{2,3,\dots,32\}$ \textit{small}, players $33,\dots,48$ \textit{medium} and players $49,\dots,64$ \textit{big}. For $s$ small, $m$ medium, and $b$ big, let the comparison matrix $P$ satisfy:
\begin{itemize}
\item $P_{1s}=1$, $P_{1m}=0.4$, $P_{1b}=0.6$,
\item $P_{sm}=P_{sb}=0$,
\item $P_{mb}=0.49$.
\end{itemize}

Consider $\sigma=(1,2,\dots,32,33,49,34,50,\dots,48,64)$ and $\sigma'=(1,\dots,64)$. Then it is straightforward to compute $\mwp(1,P)=\wp(1,P,\sigma)> 0.506$ and $\wp(1,P,\sigma')=0.502<\mwp(1,P)$ so $\sigma'$ is not optimal for player 1. However,
$\wp_{0.02}(1,P,\sigma)<0.429 < 0.432< \wp_{0.02}(1,P,\sigma').$
Indeed, for $\sigma$, the $0.02$-perturbation $P^*$ of $P$ satisfying
\begin{itemize}
\item $P^*_{1s}=0.98$, $P^*_{1m}=0.38$, $P^*_{1b}=0.58$,
\item $P^*_{sm}=P_{sb}=0$,
\item $P^*_{mb}=0.51$
\end{itemize}
yields $\wp_{0.02}(1,P,\sigma)<\wp(1,P^*,\sigma)<0.429$.

On the other hand, consider the draw $\sigma'$ and take any $P'\in\mathcal{P}(P,0.02)$. Player 1 wins his first five matches with probability at least $0.98^5$. Then he beats every medium (resp., big) player with probability at least 0.38 (resp., 0.58) and the probability that a big player wins the right half is at least $0.51-0.02=0.49$. Overall,
$\wp(1,P',\sigma)\geq 0.98^5\cdot (0.49\cdot 0.58 + 0.51\cdot 0.38)>0.432.$
\end{example}

\subsection{Computational Questions}
We extend the TFP and PTFP problem with robustness and we focus on the question of approximating the $\eps$-worst drop $d_{\eps}(i^*,P,\sigma)$.

Denote by $\xi(P)=\min\{P_{ij}\vert P_{ij}\neq 0\}$ the smallest nonzero value in the comparison matrix (if $P$ is deterministic then $\xi(P)=1$).
It is straightforward to see that for $\eps\leq \xi$, the $\eps$-guaranteed winning probability $\wp_{\eps}(i^*,P,\sigma)$ is a polynomial in $\eps$ whose constant term is $\wp(i^*,P,\sigma)$. For such $\eps$, our approximation $\wh{d}_{\eps}(i^*,P,\sigma)$ of the drop $d_{\eps}(i^*,P,\sigma)$ is the linear term of this polynomial. 
Since higher order terms of $\eps$ are ignored, the number $\wh{d}_{\eps}$ 
is really only an approximation of $d_{\eps}$. However, for all $c<N$, if $\eps$ is sufficiently
small (smaller than $cN^{-2}$), then $\wh{d}_{\eps}$ is within $\pm c\eps$ of $d_{\eps}$, so the approximation is tight. 
Formally, we consider the following problems.

\begin{problem}[RTFP] The Robust Tournament Fixing Problem (RTFP):
Given a player set $[N]$, a deterministic comparison matrix $P$,
a distinguished player $i^*\in N$ and $c \in \en$,
does there exist a draw $\sigma$ for the player set $N$ such that
(a)~player $i^*$ is the winner of $C(P,\sigma)$; and
(b)~$\wh{d}_{\eps}(i^*,P,\sigma) \leq c\eps$, for all $\eps>0$ sufficiently small?
\end{problem}

\begin{problem}[RPTFP] The Robust Probabilistic Tournament Fixing Problem (RPTFP):
Given a player set $[N]$, a comparison matrix $P$, a distinguished player $i^*\in [N]$, target probability $q\in[0,1]$ and $s \in \er$, does there exist a draw $\sigma$ for the player set $[N]$ such that 
(a)~$\wp(i^*,P,\sigma) \geq q$; and 
(b)~$\wh{d}_{\eps}(i^*,P,\sigma) \leq s\eps$, for all $\eps>0$ sufficiently small?
\end{problem}

Note that the drop computation is an infimum over all $\eps$-perturbation matrices, 
and thus represents a universal quantification. 
Thus in contrast to the TFP and the PTFP problems which have only existential 
quantification over draws, the robustness problem has a quantifier alternation of 
existential and universal quantifiers.

\section{Algorithms}

The aim of this section is to prove the following theorem.

\begin{theorem}\label{thm1}
Given a comparison matrix $P$ (deterministic or probabilistic), sufficiently small $\eps>0$, distinguished player $i^*$ and a draw $\sigma$, the value $\wh{d}_{\eps}(i^*,P,\sigma)$ can be computed in polynomial time. 
\end{theorem}

We start with a lemma stating that if we are allowed to perturb by $\eps$, perturbing by less is not worth it.

\begin{lemma}\label{perturbalot} For every $C(P,\sigma)$ with a distinguished player $i^*$ and any $\eps>0$ there exists an $\eps$-perturbation $P'$ of $P$ that gives the worst possible winning probability $\wp(i^*,P',\sigma)=\wpe(i^*,P,\sigma)$ for player $i^*$ and that for each $1\leq i\neq j\leq N$ satisfies that either $|P'_{ij}-P_{ij}|=\eps$ or $P'_{ij}\in\{0,1\}$. 
\end{lemma}
\begin{proof} For a fixed draw $\sigma$, the winning chance $\wp(i^*,P,\sigma)$ can be expressed using variables $P_{ij}$, $1\leq i\neq j\leq N$. Since each combination of outcomes of the matches determines if $i^*$ won or not, this expression is a polynomial and is linear in each of the $P_{ij}$'s. If we view $\mathcal{P}(P,\eps)$ as a subset of $\er^{\binom{N-1}{2}}$ with $\ell_\infty$ metric, then the minimum of the function $\wp(i^*,P',\sigma)$ over the set $\mathcal{P}(P,\eps)$ is attained on the boundary. Hence the result follows.
\end{proof}

\subsection{Deterministic setting}
First we focus on the deterministic case. If a player doesn't win a tournament we say he \textit{lost}.
Consider a complete deterministic tournament $C(P,\sigma)$ with winner $i^*$ and let $\eps\in(0,1)$. Lemma~\ref{perturbalot} gives Corollary~\ref{coro_1}.

\begin{corollary}\label{coro_1} 
The value $\wpe(i^*,N,P,\sigma)$ is attained for a matrix $P'$ with some $(0,1)$ matches altered into $(\eps,1-\eps)$ matches and the remaining matches left intact.
\end{corollary}

\noindent{\em Crucial matches.}
Denote by $S_\eps$ the set of $(0,1)$ matches such that altering them to $(\eps,1-\eps)$ matches gives $P'\in\mathcal{P}(P,\eps)$ satisfying $\wp(i^*,P',\sigma)=\wpe(i^*,P,\sigma)$.
A $(0,1)$ match is called \textit{crucial\/} if replacing it (and it only) by a $(1,0)$ match makes $i^*$ lose. Denote the set of crucial matches by $C$. The next lemma says that the crucial matches are the only matches that matter for the sake of determining $\wh{d}_\eps(i^*,P,\sigma)$.

\begin{lemma}\label{lemm_3} 
Suppose $C(P,\sigma)$ is won by player $i^*$. Replace some subset $S$ of $(0,1)$ matches by $(\eps,1-\eps)$ matches to get $P'\in\mathcal{P}(P,\eps)$. Suppose $c$ of those $|S|=s$ matches are crucial. Then
$\wp(i^*,P',\sigma)=1-c\cdot\eps + \eps^2Q(\eps)$,
for some polynomial $Q(\eps)$.
\end{lemma}

\begin{proof} We account for three scenarios: (A) All matches play out as before; (B) a single crucial match becomes an upset; (C) a single non-crucial match becomes an upset. For appropriate polynomials $Q_i(\eps)$, $i=1,2,3,4$, the chances of this happening are
\begin{align*}
&\textrm{(A)} \qquad (1-\eps)^s &&=1-s\cdot\eps +\eps^2Q_1(\eps)\\
&\textrm{(B)} \qquad c\eps(1-\eps)^{s-1} &&=c\eps+\eps^2Q_2(\eps)\\
&\textrm{(C)} \qquad (s-c)\eps(1-\eps)^{s-1} &&=(s-c)\eps +\eps^2Q_3(\eps)
\end{align*}
Since $i^*$ wins in (A) and (C) and loses in (B), we have
$$1-c\eps - \eps^2 Q_4(\eps)\leq \wpe(i^*,P',\sigma) \leq 1-c\eps - \eps^2Q_2(\eps)
$$
where $Q_4(\eps)=Q_1(\eps)+Q_3(\eps)$. The result follows.
\end{proof}

Lemma~\ref{lemm_3} yields the following corollary.
\begin{corollary}\label{coro_2} 
Let $C(P,\sigma)$ be a deterministic tournament with winner $i^*$ and $c$ crucial matches and let $\eps\in (0,1)$. Then $\wh{d}_\eps(i^*,P,\sigma) = c\eps$.
\end{corollary}



\noindent We now show that crucial matches can be found efficiently.

\begin{lemma} Given a complete deterministic tournament $C(P,\sigma)$, there exists an $O(N\log N)$ algorithm that computes the number of crucial matches.
\end{lemma}
\begin{proof} Going from below, at each node of the tournament (including the root) we store the winner of the corresponding subtournament. Since $N-1$ matches are played in total, this takes $O(N)$.

Now imagine a single match corresponding to a subtree $S$ was switched. In order to determine if the original winner still wins we might need to recompute the winners of all the subtournaments containing $S$. Since the depth of the tree is $\log(N)$, there are at most $\log(N)$ such tournaments and recomputing the winner of each takes constant time. Overall we get the running time $O(N+N\cdot\log N)=O(N\log N)$.
\end{proof}

This proves the deterministic case of Theorem~\ref{thm1}.

\subsection{Probabilistic setting}
The idea is similar to the deterministic case. By Lemma~\ref{perturbalot}, we want to adjust each match in one of the two directions as much as possible. Since $\wp(1,P,\sigma)$ is a polynomial linear in each $P_{ij}$, for any fixed $P_{ij}$ we can express it as $$\wp(1,P,\sigma)=\alpha_{ij}\cdot P_{ij}+\beta_{ij}.$$
The next lemma shows that in order to obtain the approximation $\wh{d}_\eps(1,P,\sigma)$ of the $\eps$-worst drop, we want to adjust $P_{ij}$ based on the sign of $\alpha_{ij}$ and that in the end we get $\wh{d}_\eps(1,P,\sigma)=\sum_{i,j} |\alpha_{ij}|$. 


\begin{lemma}\label{taylor} Let $C(P,\sigma)$ be a complete probabilistic tournament. 
 Pick $1\leq i\neq j\leq N$ and consider $\wp(1,P,\sigma)$ as a function of $P_{ij}$. Then $\wp(1,P,\sigma)=\alpha_{ij}\cdot P_{ij}+\beta_{ij}$ for suitable constants $\alpha_{ij},\beta_{ij}\in\er$ and $\wh{d}_\eps(1,P,\sigma)=\sum_{1\leq i\neq j\leq N} |\alpha_{ij}|$ for all sufficiently small $\eps>0$.
\end{lemma}
\begin{proof} Consider $\wp(1,P,\sigma)$ as a function $f\colon (0,1)^{N\choose 2}\to (0,1)$ of variables $P_{ij}$, $1\leq i \neq j\leq N$. Since each combination of outcomes of the matches determines if player 1 won or not, $f$ is a polynomial that is linear in each $P_{ij}$. Plugging in the actual values for all of them but one, we get the desired $\wp(1,P,\sigma)=\alpha_{ij}\cdot P_{ij}+\beta_{ij}$ for suitable constants $\alpha_{ij},\beta_{ij}\in\er$.

For notational convenience, assume that $P_{i,j}\not\in\{0,1\}$ for each $1\leq i\neq j\leq N$. \cref{perturbalot} then implies that for sufficiently small $\eps>0$, there is a worst $\eps$-perturbation $P'$ such that $|P'_{ij}-P_{ij}|=\eps$ for all $1\leq i \neq j\leq N$. Let
$$ e_{ij}=\begin{cases}
-1&\mbox{if } P'_{ij}-P_{ij}=\eps \\
+1&\mbox{if } P'_{ij}-P_{ij}=-\eps
\end{cases}
$$
By the multivariate version of Taylor's theorem we get
$$\wp(1,P',\sigma) = \wp(1,P,\sigma) -\eps\sum_{ij} e_{ij}\alpha_{ij} +\eps^2Q(\eps)
$$
for some polynomial $Q(\eps)$. Clearly the maximum of $\sum_{ij} e_{ij}\alpha_{ij}$ is attained if each $e_{ij}$ has the same sign as the corresponding $\alpha_{ij}$ (for $\alpha_{ij}=0$ the value $e_{ij}$ doesn't matter) and it equals $\sum_{ij} |\alpha_{ij}|$.
\end{proof}

Finally, we establish Lemma~\ref{probalgo} that covers the probabilistic case of Theorem~\ref{thm1}.

\begin{lemma}\label{probalgo} Let $C(P,\sigma)$ be a complete probabilistic tournament.
 Then $\wh{d}_\eps(1,P,\sigma)$ can be computed in polynomial time.
\end{lemma}
\begin{proof} By \cref{taylor}, for all sufficiently small $\eps>0$ we have $\wh{d}_\eps(1,P,\sigma)=\sum_{1\leq i\neq j\leq N} |\alpha_{ij}|$, where each $\alpha_{ij}$ is given by  $\wp(1,P,\sigma)=\alpha_{ij}\cdot P_{ij}+\beta_{ij}$.

There exists a recursive $O(N^2)$ algorithm that computes the winning probability of a player (see~\cite{VAS09}). This algorithm can be naturally modified to find $\alpha_{ij}$, $\beta_{ij}$ for some fixed $1\leq i\neq j\leq N$. Indeed, instead of computing the winning probability as a number we compute it as a linear function of one variable $P_{ij}$. Repeating this procedure for all $\binom{N}{2}$ parameters $P_{ij}$ gives all $\alpha_{ij}$ and in turn also $\wh{d}_\eps(1,P,\sigma)$ in $O(N^2\cdot N^2)=O(N^4)$ time.
\end{proof}


\vspace{-0.5em}
\subsection{Consequences}
\vspace{-0.2em}
Our main result has a number of important consequences which we present below.

\begin{corollary}\label{coro1}
The RTFP and the RPTFP problems are NP-complete.
\end{corollary}
\begin{proof}
The hardness result is an easy consequence of the NP-hardness of the
TFP problem~\cite{AGMM}, which is obtained as follows.
Consider the RTFP problem (which is a special case of RPTFP problem) with
$c=N+1$ with sufficiently small $\eps>0$. Then for every draw $\sigma$ we 
have $\wh{d}_{\eps}(i^*,P,\sigma) \leq c\eps$; and hence the answer to 
the RTFP problem is yes iff the answer to the TFP problem is yes. 
Hence both RTFP and RPTFP problems are NP-hard.
Since every draw is a polynomial witness, both conditions (a) and (b)
for RPFTP can be checked in polynomial time by the results of~\cite{VAS09}
and Theorem~\ref{thm1}, respectively.
\end{proof}

In~\cite{AGMM} two important special cases are considered for TFP, namely,
tournaments where there are constant number of types of players, and 
tournaments with linear ordering on players with constant number of exceptions.
For both the cases, polynomial-time algorithms are presented in~\cite{AGMM}
to enumerate all winning draws. 
In combination with Theorem~\ref{thm1} it follows that the most robust 
winning draw can also be approximated in polynomial time.

\begin{corollary}
For the two special cases of TFP from~\cite{AGMM}, the most robust 
winning draw (if a winning draw exists) can be approximated in polynomial time.
\end{corollary}

\section{Conclusion}
We studied the problem of robustness related to fixing draws in
BKT.
We presented several illuminating examples related to the robustness 
properties of optimal and near optimal draws.
We established polynomial-time algorithm to approximate the robustness of draws 
for sufficiently small $\eps>0$.
As a consequence, for the robustness of draws, 
we establish NP-completeness in general and polynomial-time algorithms
for special cases.
Interesting directions of future work are (a)~computation of
robustness when $\eps>0$ is not small;
(b)~higher order uncertainties, such as uncertainty in $\eps$ or different $\eps$'s for different matches;
and
(c)~the robustness question in other forms of tournaments.
 
 \subsection*{Acknowledgements}
This research was partly supported by Austrian Science Fund (FWF) NFN Grant No S11407-N23 (RiSE/SHiNE), Vienna Science and Technology Fund (WWTF) through project ICT15-003, ERC Start grant (279307: Graph Games), and ERC Advanced Grant (267989: QUAREM).

\bibliographystyle{plain}
\bibliography{ijcai16}

\end{document}